\documentclass[pra,aps, superscriptaddress,twocolumn,10pt]{revtex4-1}

\usepackage[latin1]{inputenc}
\usepackage[T1]{fontenc}

\usepackage[english]{babel}

\usepackage{amsmath}
\usepackage{amsfonts}
\usepackage{amssymb}
\usepackage{amsthm}
\usepackage{amsbsy,wasysym}
\usepackage{mathtools}
\usepackage{soul}

\usepackage{graphicx}
\usepackage{xcolor}
\usepackage{tikz}
\usepackage{pgffor}
\usetikzlibrary{arrows,shapes,trees,positioning}

\usepackage[framemethod=TikZ]{mdframed}
\usepackage{hyperref}

\definecolor{dred}{rgb}{.8,0.2,.2}
\definecolor{dgreen}{rgb}{.2,0.5,.2}
\definecolor{ddred}{rgb}{.8,0.5,.5}
\definecolor{dblue}{rgb}{.2,0.2,.8}

\newcommand{\bi}[0]{\begin{itemize}}
\newcommand{\ei}[0]{\end{itemize}}
\newcommand{\ben}[0]{\begin{enumerate}}
\newcommand{\een}[0]{\end{enumerate}}
\newcommand{\beq}[0]{\begin{equation}}
\newcommand{\eeq}[0]{\end{equation}}

\newcommand{\eqr}[1]{Eq.~(\ref{#1})}

\newtheorem{theorem}{Theorem}
\newtheorem{prop}{Proposition}
\newtheorem{definition}{Definition}
\newtheorem{lemma}{Lemma}
\newtheorem{protocol}{Protocol}
\newtheorem{conjecture}{Conjecture}

\newcommand{\ket}[1]{\vert{#1}\rangle}
\newcommand{\bra}[1]{\langle{#1}\vert}
\newcommand{\proj}[1]{\ket{#1}\bra{#1}}

\newcommand{\tr}{\mathrm{Tr}}

\newcommand{\supp}{\mathrm{supp}}

\newcommand{\poly}{\mathrm{poly}}

\renewcommand{\l}[0]{\langle}

\renewcommand{\r}[0]{\rangle}

\newcommand{\av}[1]{\langle #1\rangle}

\newcommand{\abs}[1]{\lvert #1\rvert}
\newcommand{\Abs}[1]{\left\lvert #1\right\rvert}
\newcommand{\norm}[1]{\lVert#1\rVert}

\newcommand{\set}[1]{\{ #1  \}}

\newcommand{\ee}{\mathrm{e}}
\newcommand{\ii}{\mathrm{i}}

\newcommand{\Pb}{\mathbb{P}}

\RequirePackage[loading]{tracefnt}

\usepackage{times}
 
\makeatletter 
\hypersetup{pdftitle = {Direct certification of a class of quantum simulations},
	     pdfauthor = {Dominik Hangleiter, Martin Kliesch, Martin Schwarz, Jens Eisert},
	     pdfsubject = {Quantum many body physics, quantum information},
	     pdfkeywords = {Quantum simulation,
		     verification, 
		     frustration free Hamiltonian, 
		     local Hamiltonian, 
		     ground state, 
		     fidelity estimation, 
		     weak membership, 
		     IQP,
		     quantum supremacy
		     }
	    }
\makeatother

\begin{document}
\title{Direct certification of a class of quantum simulations}
\author{D.\ Hangleiter}

\affiliation{Dahlem Center for Complex Quantum Systems, Freie Universit{\"a}t Berlin, D-14195 Berlin, Germany}
\affiliation{Fakult\"at f\"ur Physik, Ludwig-Maximilians-Universit\"at, 80799
M\"unchen, Germany}

\author{M. Kliesch}
\affiliation{Dahlem Center for Complex Quantum Systems, Freie Universit{\"a}t Berlin, D-14195 Berlin, Germany}

\author{M. Schwarz}
\affiliation{Dahlem Center for Complex Quantum Systems, Freie Universit{\"a}t Berlin, D-14195 Berlin, Germany}
\author{J. Eisert}
\affiliation{Dahlem Center for Complex Quantum Systems, Freie Universit{\"a}t Berlin, D-14195 Berlin, Germany}

\date{\today}

\begin{abstract}
One of the main challenges in the field of quantum simulation and computation is to identify ways to certify the correct functioning of a device when a classical efficient simulation is not available. 
Important cases are situations in which one cannot classically calculate local expectation values of state preparations efficiently.  
In this work, we develop weak-membership formulations of the certification of ground state preparations. 
We provide a non-interactive protocol for certifying ground states of frustration-free Hamiltonians based on simple energy measurements of local Hamiltonian terms. 
This certification protocol can be applied to
classically intractable analog quantum simulations: 
For example, using Feynman-Kitaev Hamiltonians, one can 
encode universal quantum computation in such ground states. 
Moreover, our certification protocol is applicable to ground states encodings of IQP circuits demonstration of quantum supremacy. These can be certified efficiently when the error is polynomially bounded. 
\end{abstract}

\maketitle



\section{Introduction}
\label{s:intro}

Quantum devices offer the promise to outperform classical machines computationally and to solve
problems intractable in the classical domain. This idea is most prominent in the idea of a fully-fletched
universal quantum computer that could---once physically realised for large system sizes---solve problems 
like factoring in polynomial time \cite{Shor-1994}. Indeed, much effort is dedicated to realising the building blocks from
which such a universal quantum computer could be built. What is already available with present technology,
however, are quantum simulators \cite{CiracZollerSimulation}, machines able to simulate 
other quantum systems, based on an idea going back to Richard Feynman \cite{Feynman}. In particular, this is true for so-called analog quantum simulators that involve a very large number
of quantum constituents whose time evolution is not discretised as is the case
in digital
simulations. 
Prominent architectures for analog simulation
are based on cold atoms in optical lattices \cite{BlochSimulation}, trapped ions
 \cite{IonSimulation} or superconducting qubits \cite{SuperconductingSimulation}. 
As of yet, their computational power is far from clear, though.  
They are special purpose machines, in that 
one can achieve a tremendous degree of control, but only over some parameters: In this way, it is perfectly
conceivable to prepare ground states of interacting Hamiltonians to probe
effects such as high-$T_c$
superconductivity \cite{ColdFermions} or lattice gauge theories \cite{LatticeGauge2}. 
Equally well, properties of time evolution can be probed in 
dynamical quantum simulators \cite{1408.5148,Trotzky_etal12}. 

In all such approaches the problem of certification arises, providing an answer to the question: 
\emph{Is the device working precisely in the anticipated fashion?} 
In case of a quantum computation that solves a
problem in \textsf{NP}, one can of course in retrospect efficiently certify the correctness
of the solution to the decision problem by classical means. This approach is insufficient in many instances, however. To start with,
not all problems interesting to be solved on a quantum device are decision problems. 
A particularly pressing question for quantum simulations, is to find out whether one has achieved a desired
ground state preparation accurately. 
Moreover, it would be desirable to also obtain information about intermediate steps of a quantum computation or simulation. 
If parameter regimes are available that allow for an efficient classical simulation---as is the case, e.g., 
in non-interacting physical systems---one can set benchmarks in these regimes.
In the physics literature, it is surprisingly common to assume that the only way to certify the correctness
of a quantum simulation is to find regimes in which a classical simulation is available. 

In this work, we show the converse: 
Ground state preparations of
frustration-free Hamiltonians as an example of a type of quantum simulation can be
certified efficiently using local measurements only, while at the same time
the actual outcome of the simulation or quantum computation cannot be predicted
classically. 
Our certification protocol is surprisingly simple and can be applied in a range
of paradigmatic settings including the certification of IQP circuits
\cite{bremner2015average}, the toric code \cite{ToricCode}, graph states \cite{Graphs}, 
and even arbitrary quantum computations encoded in Feynman-Kitaev-type
Hamiltonians \cite{Gosset-PRL-2015, Feynman-1986,Kitaev-2002,Aharonov-2002}.
With the rapid experimental advances in the aforementioned experimental
platforms, the large-scale realization of these models has come into reach. Our
protocol yields an experimentally viable method to certify their correctness
using experimental techniques that are available now. 

The central idea of this work is to combine 
fidelity estimation techniques from quantum state tomography \cite{Cramer-NatComm-2010}
with computational models from 
quantum complexity theory \cite{Kitaev-2002,Aharonov-2002,Gosset-PRL-2015}. 
Its simplicity is the main virtue of our protocol -- both on a conceptual level and in terms of the experimental
capacities required for its implementation. 
Our main theoretical insight -- that certification may well be
possible even when classical simulation is not available -- underpins the
potential of quantum simulators as well as the necessity of their certification, 
hence addressing a timely
topic in the quest for identifying quantum simulators with superior computational power. 
An important part of the purpose of this work is to provide a 
link between the physics and computer science literature, hence contributing to
the debate on how to
reliably certify quantum simulators outperforming classical machines.

Technically, we extend a known fidelity bound for the ground states of
frustration-free Hamiltonians \cite{Cramer-NatComm-2010}.  
Formulating the problem of fidelity estimation for ground states of
frustration-free Hamiltonians as a weak-membership 
certification protocol \cite{Aolita-NatComm-2015} then allows us to address the
problem of rigorously certifying such ground states given an estimate of their
energy only.
To demonstrate that in this very setting, computationally intractable problems
can be solved, we make use
of the framework of a physically plausible variant of a Feynman-Kitaev
Hamiltonian construction \cite{Gosset-PRL-2015}. 
Thus, small technical improvements on known results allow us to show our main
result as formalised in Propositions~\ref{certification} and \ref{complexity}
and make it applicable to real-world experimental settings. 

An intriguing application of our certification protocol arises from Ref.\ \cite{bremner2015average}. In this work, the authors present a restricted (non-universal) family of \emph{commuting} quantum circuits, where sampling from their output probability distribution with an error in trace distance of less than $1/192$ is \emph{provably hard} for any classical computer up to a plausible complexity-theoretic conjecture,
unless the Polynomial Hierarchy (a generalisation of $\mathsf{P} \neq \mathsf{NP}$) collapses.
This circuit can be mapped to the ground state of a frustration-free local Hamiltonian, where our certification method can be applied. 
Let us first note that from energy measurements alone one can in general not guarantee to accept all states within a constant error. 
For such a guarantee a polynomial bound on the preparation error is required. 
However, also for a constant allowed error, if the certification accepts a state, then the state is indeed within the required error bound. 
In this sense, the combination of our methods with the results of Ref.\ \cite{bremner2015average} yield an experimental proposal to demonstrate
\emph{quantum supremacy} \cite{preskill2013quantum}.

\subsection{Related work}
In computer science, notions of verified quantum computation \cite{AharonovInteractive,Kashefi,FitzsimonsNew,Aolita-NatComm-2015}
address the question of how the correct functioning of quantum machines can be checked. 
In such approaches, it is usually a small, well-characterised quantum system that is used to 
certify the functioning of a larger quantum machine making use of the idea of
interactive proof systems involving several rounds of interaction between a
restricted verifier and a powerful prover. In multi-round systems, the dynamics of a system can even be certified
without the need for any quantum mechanical capacities
\cite{Reichardt-2012,Reichardt-Nature-2013}. 
In the same spirit, it is shown in Refs.~\cite{Wiebe-PRL-2014,Wiebe-PRA-2014,Wiebe-NJP-2015} how trust in a 
quantum simulator can be amplified to certifying an untrusted simulator or to learning about the Hamiltonian of an unknown system.
Very recently, a single-round scheme was devised \cite{FitzsimonsNew} that
relies on preparing the ground state of a Feynman-Kitaev Hamiltonian. Here, the
authors consider the general situation in which casting verification
as a \textsf{QMA} problem does not directly lead to a verification protocol, since the
task of the verifier ``may well be comparable in complexity to independently
implementing the circuit to be evaluated''. For this general case the authors
devise a rather sophisticated protocol, which, to an extent, seems to be skew
to the conceptual simplicity desired of a quantum simulator. 

\subsection{Notions of quantum simulation}
\label{s:qs}

We start our discussion by emphasising what properties of quantum simulators we consider indispensable.
Except from solving an interesting problem from the physics perspective, a
quantum simulator should satisfy the following two criteria: it should be
certifiable (reliable), and it should be able to solve a computationally hard
problem that cannot efficiently be solved using classical means (proper) (cf.\
Ref.~\cite{Hauke-RPP-2012}.  
\begin{itemize}
\item[(i)] A quantum simulator should be \emph{proper}. We call a \emph{ground
state quantum simulator} (GS-QS) proper if using it one can solve a
computational task that is believed to be intractable on a classical computer.
 
\item[(ii)] A quantum simulator should be \emph{reliable}. We call a GS-QS \emph{reliable} if there exists a weak-membership certification test (Def.~\ref{weak membership test}) by which the actual preparation $\rho_p$ of the ground state can be certified. 
\end{itemize}

\subsection{Definitions}
\label{s:definitions}

To begin with, let us define the key concepts for what follows. In particular,
these concepts are properties of Hamiltonians relevant
for ground state quantum simulations considered here, related to being local, frustration-free and having a Hamiltonian gap.

\begin{definition}[$k$-local Hamiltonians] \label{local H}
Let $\Lambda$ be a finite set of sites, and
associate with each site $v \in \Lambda$ a finite-dimensional Hilbert
space $\mathcal{H}_v$. A Hamiltonian $H = \sum_{\lambda \subset \Lambda} h_\lambda$ defined on
$\mathcal{H} = \bigotimes_{v \in \Lambda} \mathcal{H}_v$ is called
($k$-)local if each term $h_\lambda$ acts on at most $k = O(1)$ sites within
$\lambda$, i.e., has finite support $|\supp(h_\lambda)| \leq
k$ independent of the system size $N = |\Lambda|$. 
\end{definition}
Throughout, we assume the local terms to be bounded in spectral norm,
$\norm{h_\lambda} \le O(1)$ for all $\lambda$ independent of the system size. Most important will be the situation of encountering spatial locality, where the supports of the $h_\lambda$ are connected sets on a physical lattice such as a cubic lattice.

A Hamiltonian $H$ is said to be \emph{(polynomially) gapped} if the energy gap $\Delta$ between the ground state energy $E_0$ and the first excited eigenenergy $E_1$ satisfies $\Delta = \Delta(N) \propto 1/\poly(N)$.
An important case where $E_0$ can be easily obtained is the case of frustration-free Hamiltonians. Hence, we can set $E_0=0$ in the following definition. 

\begin{definition}[Frustration-free Hamiltonians] \label{frustration-free}
A local Hamiltonian $H = \sum_\lambda h_\lambda$ with ground-state energy $E_0
= 0$ is called
\emph{frustration-free (FF)}, if the projection $P_0$ onto the ground-state space
of $H$ satisfies $H P_0 = 0$ and $h_\lambda P_0 = 0$ for all $\lambda$.
\end{definition}

\section{Certification of proper ground state quantum simulations}
\label{s:certificate}

We now turn to the actual certification scheme and the main insight of this work:
Quantum systems that are described by spatially local, frustration-free and
gapped Hamiltonians with unique ground state realise (i) reliable GS-QS some of
which are even (ii) proper. In other words, in a realistic physical setting,
namely, for systems described by spatially local frustration-free gapped Hamiltonians
with unique ground state it is indeed possible to solve problems that are
believed to be intractable classically.
This insight is very much in the original spirit of quantum simulators presumably outperforming classical devices.

To prove that GS-QS with frustration-free and polynomially gapped Hamiltonians
with unique ground state is reliable we use that the ground state of such
Hamiltonians be certified efficiently by an energy measurement only
\cite{Cramer-NatComm-2010}. Considering the case in which this energy
measurement is performed by measuring local Hamiltonian terms 
\cite{Kitaev-2002,Aharonov-2002}, we formalise this in
Proposition~\ref{certification}. There, we provide bounds on the required number of measurements given a desired accuracy. 

To demonstrate that in the same setting also computationally hard problems
can presumably be solved, we make use of the fact that universal quantum
computation can be performed by cooling to ground states of frustration-free
Hamiltonians \cite{Bravyi-siam-2010,Gosset-PRL-2015}. To demonstrate this one
needs to show that any quantum computation can be reduced to (local)
measurements on that ground state. Specifically, Proposition~\ref{complexity} states
that adiabatic quantum computation can be performed efficiently using a
polynomially gapped spatially local frustration-free Hamiltonian
$H_{\text{ac}}$. Since adiabatic computation (ac) is universal \cite{Aharonov-SIAM-2008}, any problem in \textsf{BQP} can be solved using the system described by $H_{\text{ac}}$.

\subsection{Weak-membership certification}
\label{s:weak membership}

We now define weak-membership quantum state certification. This is at the heart our notion
of certifying the correctness of a preparation. It captures precisely what it means to ``have prepared a
state to a given accuracy in the laboratory''. Is is conveniently conceptualised in terms of a game between a sceptic and restricted certifier Arthur and a powerful but untrusted quantum prover Merlin. Arthur has access to classical computation and spatially local measurements; Merlin is able to prepare arbitrary quantum states
of large quantum systems. The input to the game is the classical description of
the Hamiltonian $H$ the ground state of which we want to prepare. Merlin
prepares a number of independent and identical copies of a quantum state
$\rho_p$. Arthur's goal is to certify that $\rho_p$ is indeed the ground state
of $H$. In our setting, the only interaction we require between Arthur and Merlin is the local measurements of Arthur on the state preparation given to him by Merlin \cite{Aolita-NatComm-2015}. 
The test accepts if the state $\rho_p$ prepared by Merlin is sufficiently close to the 
actual ground state $\rho_0$ in fidelity that is given by $F(\sigma, \rho) :=
\tr[\rho\sigma]$ for at least one of $\sigma$ and $\rho$ pure. The fidelity is
related to the trace distance $d(\sigma, \rho) = \tr[\abs{\sigma - \rho}]/2$
via $1- F(\sigma, \rho)^{1/2} \leq d(\sigma, \rho) \leq ({1- F(\sigma,
\rho)})^{1/2}$. A robust reading of what it precisely means for Arthur to
``certify'' $\rho_p$ as a good approximation of $\rho_0$ in terms of the
fidelity is the notion of weak-membership state certification, which we
illustrate in Fig.~\ref{weak membership}.

\begin{definition}[Weak-membership quantum state certification] \label{weak membership test}
Let $F_T >0$ be a threshold fidelity and $0 < \alpha < 1$ a maximal failure probability. 
A test which takes as an input a classical description of $\rho_0$ and copies of a preparation of $\rho_p$, and outputs ``reject'' or ``accept'' is a \emph{weak-membership certification test} if with high probability $\geq 1- \alpha$ it rejects every $\rho_p$ for which  $F(\rho_p, \rho_0) \leq F_T$, and accepts every $\rho_p$ for which $F(\rho_p,\rho_0) \geq F_T + \delta$ for some fidelity gap $\delta > 0$.
\end{definition}

\begin{figure}[t!]
\begin{tikzpicture}[scale = 1.5]
  \node (center) at (0,0)  {}; 

  \draw[fill = red!60] (center) ellipse (2 and 1.5);
  \draw[fill = black!40!blue!50,dashed] (center) ellipse (1.25 and 1.25*.75) ;
  \draw[fill = green,dashed] (center) ellipse (.75 and .75*.75) ;

  \draw [<->,thick ](center) --node[sloped,above]  {$1-F_T$}  (1.25,0) ;
  \draw [<->, ](-.75,0) - - node[sloped,below]  {$\delta$}  (-1.25,0) ;
  \filldraw [black!70] (center) circle[radius=2pt] node (rho0) {} node [left,black] {$\rho_0$}; 
  \filldraw [black!70] (-1.25,.8) circle[radius=2pt] node (rhop) {} node [left,black] {$\rho_p$}; 
  \draw [dotted] (center)  --node [midway,sloped,below] {$1-F$} (rhop);
  \node (reject) at (0,1.15) {reject};
  \node at (0,-.2) [below] {accept};
  \node at (.8,-.3) [below] {\textbf{?}};
\end{tikzpicture}
\caption{\textit{Weak-membership certification}: When performing
protocol~\ref{certprotocol} it turns out that a prepared state $\rho_p$ is accepted if 
$F \geq F_T +\delta$ (completeness) and rejected if $F < 1-F_T$ (soundness). 
If $\rho_p$ lies in the blue region, no acceptance is guaranteed. However, in case of acceptance, $\rho_p$ is within the desired fidelity region. }
\label{weak membership}  
\end{figure}
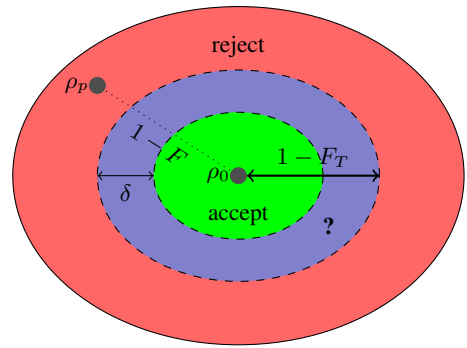  

\subsection{Certification protocol}
\label{s:protocol}

In this section, we lay out the actual certification protocol of frustration-free Hamiltonians. The anticipated state to be certified
is the ground state $\rho_0$ of $H$, on which a sequence of local measurements 
is performed,
followed by efficient classical post-processing.

\begin{protocol}[Certification protocol for FF Hamiltonians]\label{certprotocol} $\,$
\begin{enumerate}
\item Arthur chooses a threshold fidelity $F_T< 1$, maximal failure probability $1> \alpha >0$ and an estimation error $\epsilon \leq (1-F_T)/2$.

\item Now Arthur asks Merlin to prepare a sufficient number of copies of the
  ground state $\rho_0$ of the frustration-free and gapped Hamiltonian $H = \sum_{\lambda} h_\lambda$, i.e., the unique state that fulfils $H \rho_0 =0$. 

\item Arthur measures each of the $n$ local terms $m$ times on the copies of
the state $\rho_p$ prepared by Merlin to determine an estimate $E^*$ of the
expectation value $\sum_\lambda \tr[\rho_p h_\lambda]$. Each measurement is
performed on a \emph{new} copy, and $m$ is given by expression \eqref{measurements}.

\item From the estimate $E^*$ he obtains an estimate $F_{\min}^*$ of the
lower bound $F_{\min}$ \eqref{fidelity bound} on the fidelity $F =F(\rho_p, \rho_0)$
such that $F_{\min}^* \in [F_{\min} - \epsilon, F_{\min} + \epsilon]$ with probability at least $1-\alpha$. 

\item  If $F_{\min}^* < F_T + \epsilon$ he rejects, otherwise he
accepts. 
\end{enumerate}
\end{protocol}

The subsequent proposition makes the resources needed for such a certification
more precise, and sets the fidelity gap $\Delta$ for which
Protocol~\ref{certprotocol} is a weak-membership test. 

\begin{prop}[Weak-membership certification] \label{certification}
Let $H=\sum_\lambda h_\lambda$ be a gapped Hamiltonian with unique ground state,  ground state energy $E_0 \coloneqq 0$, gap $\Delta$, and interaction strength $J  = \max_\lambda \norm{h_\lambda}$. 
Then Protocol~\ref{certprotocol} is a weak-membership certification test with fidelity gap
\begin{align} 
\delta =  \left( 1- F_T \right) \left( 1- \frac{\Delta}{\norm{H}} \right) +
\frac{2 \epsilon \Delta}{\norm{H}} 
\, , \label{fidelity gap}
\end{align}
 and requires
\begin{align}
  m \geq  \frac{J^2 n^2 }{2 \Delta^2 \epsilon^2 } \ln \left[ - \frac{n +1}{\ln
 (1 -\alpha)} \right]\,. 
 \label{eq:sample complexity}
\end{align}
measurements of each of the $n$ local terms on $\rho_p$ to determine the expectation value $\av{H}_{\rho_p} =\tr[ H \rho_p]$. 
\end{prop}

With protocol~\ref{certprotocol} one is therefore able to efficiently certify ground state preparations
of polynomially gapped Hamiltonians $H$ that
are at least $1/\poly(n)$ close to the target state in fidelity. 

To keep the notation simple, we have assumed $E_0 = 0$. However, the general case with arbitrary $E_0$ can be reduced to it whenever the value of $E_0$ is known apriori. 
For instance, this is the case whenever $H$ is frustration-free, since the global ground state energy can be obtained from the local ones. 

Moreover, in order to obtain the estimate $F^*_{\min{}}$ of the fidelity lower bound $F_{\min{}}$ the value of the gap $\Delta$ needs to be known. 

Note that \textit{no} assumptions on the prepared state $\rho_p$ are made, in particular, it need not be pure. 

Note also that it does not matter how the measurement of the total energy $\l H \r$ is performed; 
depending on the setup at hand, it may be more suitable to measure the energy directly rather 
than measure all terms individually as insinuated in Step 3 of the protocol. 
One major advantage of our approach is that one can perform the output measurement on the 
same copies as the certification measurement. This means that our certification protocol can 
be simply integrated in the readout protocol of a GS-QS: perform the certification measurements 
on the copies of $\rho_p$ first and then use the same states for the readout measurements if 
the certification test accepts $\rho_p$.

\subsection{Proof of Proposition \ref{certification}}
\label{s:proof}

In order to show Proposition~\ref{certification}, we require some detail on how to estimate the global energy from measurements of local Hamiltonian 
terms, cast into the form of a large-deviation bound. Such a bound is given by
the following Lemma, which is stated and proved along the same lines as
Lemma~S4 in Ref.~\cite{Aolita-NatComm-2015}. 

\begin{lemma}[Estimation of the energy] Decompose the local terms in their
eigenbasis: $h_\lambda = \sum_\mu e_{\lambda,\mu} P_{\lambda,\mu}$, where the
$P_{\lambda,\mu}$ are orthogonal projections onto the eigenspaces corresponding
to eigenvalues $e_{\lambda,\mu}$. Let $X_\lambda^{(i)}$ be the random variable that takes the value $e_{\lambda,\mu}$ with probability $\tr (\rho_p P_{\lambda,\mu})$ by the measurement of $h_\lambda$ on the $i^{\text{th}}$ copy of $\rho_p$. Moreover, let 
\begin{equation}
\av{h_\lambda}_{\rho_p}^* = \frac{1}{m} \sum_{i=1}^m X_\lambda^{(i)} 
\end{equation}
be the estimate
of $\av{h_\lambda}$ on $\rho_p$ by a finite-sample average of $m$ measurement
outcomes, and $\av{H}_{\rho_p}^* = \sum_\lambda \av{h_\lambda}_{\rho_p}^*$ the resulting estimate of $\av{H}_{\rho_p}$.
As above, define $J = \max_\lambda \norm{ h_\lambda}$. For $\bar{\alpha} \geq
1/2$ it holds that 
\begin{equation}
\Pb \left[\abs{\av{H}_{\rho_p}^* - \av{H}_{\rho_p} } \leq \epsilon \right] \geq \bar{\alpha} \, ,
\end{equation}
whenever 
\begin{align}
 m \geq  \frac{J^2 n^2 }{2 \epsilon^2 } \ln \left[ \frac{n +1}{\ln
 (1/\bar{\alpha})} \right]\,. \label{measurements}
\end{align}
\end{lemma}

\begin{proof}

By Hoeffding's inequality,
\begin{align} \forall \lambda \in [n]: \, \Pb \left[  \abs{\av{h_\lambda}_{\rho_p}^* -  
\av{h_\lambda}_{\rho_p} } \geq \frac{\epsilon}{m} \right] \leq 2 \ee^{- 2
\epsilon^2/ m
\norm{h_\lambda}^2} \, , \label{eq:hoeffding1}
\end{align}
since the $\set{X_\lambda^{(i)}}_\lambda$ are independent random variables and
$0 \leq X_\lambda^{(i)} \leq \norm{h_\lambda}$. Eq.~\eqref{eq:hoeffding1} is equivalent to 
\begin{align}
 \forall \lambda \in [n]: \, \Pb \left[  \abs{\av{h_\lambda}_{\rho_p}^* -
 \av{h_\lambda}_{\rho_p} } \leq \epsilon \right] \geq 1 - 2 \ee^{-2
 m \epsilon^2/ \norm{h_\lambda}^2} \, , 
\end{align}
and since all measurements are independent
\begin{align}
 \Pb & \left[  \forall \lambda \in [n]: \, \abs{\av{h_\lambda}_{\rho_p}^* -
 \av{h_\lambda}_{\rho_p} } \leq \epsilon \right] \\
 & \geq \prod_\lambda \left( 1 - 2 \ee^{-\frac{ 2 m
 \epsilon^2}{\norm{h_\lambda}^2}} \right)  \geq \left( 1 - 2 \ee^{- 2 m
 \epsilon ^2/ J^2} \right )^n \, , 
\end{align}
resulting in 
\begin{align} \Pb \left[\abs{\av{H}_{\rho_p}^* - \av{H}_{\rho_p} } \leq
\epsilon \right] \geq \left( 1 - 2 \ee^{- 2 m \epsilon^2 / n^2 J^2} \right )^n
\geq \bar{\alpha} \, . \label{eq:hoeffding2}
\end{align}
Eq.~\eqref{eq:hoeffding2} holds whenever
\begin{align}
    m \geq \frac{J^2 n^2 }{2 \epsilon^2} \ln\left[ \frac{2}{1 -
    \bar{\alpha}^{1/n} }\right] \eqqcolon m_{\mathrm{opt}}\, , 
\end{align}
where we can upper-bound $m_{\mathrm{opt}}$ as $m_{\mathrm{opt}} \leq (J^2 n^2/2
\epsilon^2) \ln[(n+1)/\ln(1/\bar{\alpha}]
)$ \cite{Aolita-NatComm-2015}. This shows the claim.  
\end{proof}

We would like to bound the error between the actual fidelity $F(\rho_0,\rho_p)$
and the estimate thereof given the energy measurement $\l H \r^* _{\rho_p}$. To
that end, we first upper- and lower-bound the fidelity by the expectation value
$\l H\r_{\rho_p}$ thereby complementing the lower bound from
Ref.~\cite{Cramer-NatComm-2010}.

\begin{lemma}[Bounds on the fidelity \cite{Cramer-NatComm-2010}] 
Given a Hamiltonian $H$ with gap $\Delta  > 0$ above the
\textit{unique} ground state $\rho_0$ with energy $E_0 = 0$, and maximum energy
$E_{\max} = \norm{H}$. Suppose $\av{H}_{\rho} < E_1$, the energy of the first excited state, for some prepared state $\rho$. Then the overlap $F(\rho_0,\rho)$ of $\rho$ with the ground state is bounded as
\begin{align} 
\label{fidelity bound}
F_{\max}(\rho) \coloneqq 1- \frac{\av{H}_\rho}{\norm{H}} \ge F(\rho_0, \rho)
\ge 1- \frac{\av{H}_\rho}{\Delta} \eqqcolon F_{\min}(\rho)\, . 
\end{align}
\end{lemma}

\begin{proof} 
Decompose $H = \sum_{i=0}^n E_i \proj{E_i}$. Then, 
\begin{align}
\norm{H} (1 - \tr[\rho_0 \rho]) & = E_{\max} \sum_{i>0} \l E_i, \rho E_i \r \\
& \ge \tr [H\rho]\geq \Delta \sum_{i>0} \l E_i, \rho E_i \r \nonumber\\
& = \Delta (1 - \tr[\rho_0 \rho])\, ,\nonumber 
\end{align}
(equivalently $\norm{H} (1- P_0) \geq H \geq \Delta(1-P_0)$, where $P_0$ projects onto the ground space) yielding $ 1- \tr[H\rho]/\norm{H} \geq F(\rho_0, \rho) \geq 1- \tr[H\rho]/\Delta$. 
\end{proof}

This means that the lower bound to the fidelity resulting from those estimates $F_{\min}^* := F_{\min}(\rho_p)^{*} = 1- \av{H}_{\rho_p}^*/\Delta$ obeys
\begin{align} 
\Pb \left[\Abs{F_{\min}^* - F_{\min} } \leq \epsilon \right] \geq 1- \alpha \, , \label{bound1}
\end{align}
with $F_{\min} = F_{\min}(\rho_p)$ whenever 
\begin{align}
 m \geq  \frac{J^2 n^2 }{2 \Delta^2 \epsilon^2 } \ln \left[ - \frac{n +1}{\ln
 (1 -\alpha)} \right]\,. 
\end{align}

\begin{proof}[Proof of Proposition~\ref{certification}]
(i) (Completeness) Let $\rho_p$ be such that $F \geq F_T + \delta$ with $\delta$ given in
Eq.~\eqref{fidelity gap}. 
Then 
\begin{align} 
F \geq F_T  + \left( 1- F_T \right) \left( 1- \frac{\Delta}{\norm{H}} \right) +
\frac{2 \epsilon \Delta}{\norm{H}} \, , 
\end{align}
which is equivalent to 
\begin{align} 
\label{bound2}
F_T + 2 \epsilon \leq 1- \norm{H} (1- F) /\Delta \, . 
\end{align} 

On the other hand, the upper bound in Eq.~\eqref{fidelity bound} implies $\av{H}_{\rho_p} \leq
\norm{H} (1 - F )$ and therefore $F_{\min} = 1- \av{H}_{\rho_p}/\Delta
\geq  1 - \norm{H} (1- F) /\Delta$. From this, together with the bound
\eqref{bound2}, we
obtain $F_T + 2 \epsilon\leq F_{\min} $.

Finally, it follows from \eqr{bound1} that 
$\Pb[\abs{F_{\min} - F_{\min}^*} \leq \epsilon ]\geq 1-\alpha$, so that
\begin{align} 
\Pb[F_{\min}^* \geq F_T + \epsilon ]\geq 1-\alpha \, . 
\end{align} 
Hence, with probability larger than $1-\alpha$ the test described in
protocol~\ref{certprotocol} accepts $\rho_p$.

(ii) (Soundness) Let $\rho_p$ be such that $F < F_T$. It then follows that
$F_{\min} \leq F < F_T$. Hence, 
\begin{equation}
\Pb[F_{\min}^* < F_T + \epsilon] >1-\alpha 
\end{equation} 
is implied by \eqr{bound1} so that $\rho_p$ is rejected with probability at least $1-\alpha$.
 \end{proof}

\subsection{Use of the bounds in analog simulation}
\label{s:analog simulation}

The above techniques can readily be applied to frustration-free analog GS-QS. Indeed, many interesting Hamiltonians can be identified as 
frustration-free Hamiltonians in the above sense. Notably, for one-dimensional
quantum systems, the parent-Hamiltonians of matrix-product
states are frustration-free \cite{MPSReps}. For two-dimensional systems, the toric code Hamiltonian \cite{ToricCode} 
constitutes an important example of
a Hamiltonian of this type.

This motivates the question of how such ground states of frustration-free Hamiltonians can be efficiently prepared.
In the next section we discuss the preparation of ground states using an
\emph{adiabatic algorithm}. Depending on the setting other methods may
prove more practical, though. For example, one can prepare the ground state by
\emph{local unitary circuit} or by \emph{engineering dissipation}. In both of
these cases, however, it is
not yet fully known how to efficiently prepare ground states, unless additional conditions are satisfied. Specifically, if all local Hamiltonian terms commute, the ground state of a frustration-free Hamiltonian can be prepared efficiently using a unitary variant \cite{Schwarz-2013} of Moser's algorithm for solving certain classical satisfiability problems \cite{Moser-2009}. Relatedly, in Ref.\ \cite{Verstraete-NatPhys-2009} the authors demonstrate that some frustration-free Hamiltonians can be cooled into their ground states by engineering dissipation.
In this scheme, the (possibly regrouped) local terms are individually cooled by engineering local dissipative terms suitably. If cooling one term cannot create an excitation at a different location the protocol is efficient. 

\subsection{Encoding quantum computations in ground states} \label{gs complexity} 
\label{s:qc}

Protocol~\ref{certprotocol} offers a convenient way to certify ground states of
frustration-free Hamiltonians. To serve as proper quantum simulators, the
problems solvable by means of these systems should be computationally hard,
i.e., intractable on a classical computer. This is the case; indeed, one can
perform universal (adiabatic) quantum computation by preparing the ground states of certain
spatially local frustration-free Hamiltonian.

\begin{prop}[Encoding quantum computation in ground states \cite{Gosset-PRL-2015}] \label{complexity} 
Let $U = U_L \cdots U_1$ be an arbitrary quantum circuit involving one- and
two-qubit unitaries $U_i$, $i = 1, \ldots, L$, that acts on $K$ qubits, and $L
= \poly(K)$, let $\ket{\phi_0} \in (\mathbb{C}^2)^{\otimes K}$ be an initial
state, and set $n = O(K^2)$. 

Then there exists a Hamiltonian $H_{\text{ac}}(\lambda), \, 0 \leq \lambda \leq 1$
on an $n \times n$ square lattice each site of which can be occupied by either
zero particles or one spin-1/2 particle. For the computation, a string of $2n$
particles move through the lattice. In particular, the following holds: For all $\lambda \in [0,1]$,
$H_{\text{ac}}(\lambda)$ is polynomially gapped as $\Omega(n^{-3})$. Moreover,
the evolved Hamiltonian $H_{\text{ac}}(1)$ is frustration free and has a unique ground state
$\ket{	\text{gs}}_{\text{ac}}$.
The probability of measuring the position of the particles on $\ket{\text{gs}}_{\text{ac}}$ 
such that one obtains the output state $(U_L \cdots
U_1)\ket{\phi_0}$ in their spin degrees of freedom is lower bounded by a constant. 
\end{prop}
 
In the proof of Ref.\ \cite{Gosset-PRL-2015} $H_{\text{ac}}$ is constructed as
a variant of the Feynman-Kitaev Hamiltonian. Specifically, the Hamiltonian
$H_{\text{ac}}$ describes the diffusion of a string of the particles on the lattice, where at
each plaquette two particles interact via a four-body interaction term. That
is, the position of the particle string serves as a clock register, their internal state as the work
register. Proposition~\ref{complexity} entails that (a) the ground state of the Hamiltonian
$H_{\text{ac}}(1)$ is unique, (2) can be 
efficiently prepared by preparing the ground state of $H_{\text{ac}}(0)$ and
then tuning $\lambda$ adiabatically from 0 to 1, and (c) is
polynomially gapped. Finally, (d) the output of the computation can be obtained
by performing measurements on the spin degrees of freedom. Thus, the ground state
of $H_{\text{ac}}(1)$ can be certified by an energy measurement on
$\ket{\text{gs}}_{\text{ac}}$ according
to Protocol~\ref{certprotocol}. More recently, it has been shown how to improve
this construction to 
make only use of nearest-neighbour interactions~\cite{Lloyd-2015}. 

In order to familiarise the reader with the construction, here we sketch the
most basic variant \cite{Bravyi-siam-2010} of the Feynman-Kitaev clock
construction \cite{Feynman-1986,Kitaev-2002,Aharonov-2002}. In this
construction, the discrete evolution under a quantum circuit is mapped to measuring the ground state of a local frustration-free Hamiltonian called the 
Feynman-Kitaev Hamiltonian. The discretised time evolution $(U_L \cdots U_2 U_1)\ket{\phi_0}_w$ of a $k$-qubit input state vector $\ket{\phi_0}_w$ 
is in this setting mapped onto a so-called \emph{history state} 
\begin{align}
\ket{\psi_{\text{hist}}} = \frac{1}{\sqrt{L + 1}} \sum_{t = 0 } ^L (U_t \cdots U_0 \ket{\phi_0}_w ) \otimes \ket{t}_c\, ,
\end{align}
of a bipartite system comprising a work and a clock register. The work register consists of $k$ qubits, the clock-register is an $L$-dimensional quantum system with state vectors $\{\ket{t}\}_{t = 0}^L$. The combined system is equipped with a Hilbert space $\mathcal{H}_{\text{work}}\otimes \mathcal{H}_{\text{clock}} \cong (\mathbb{C}^2)^{\otimes k} \otimes \mathbb{C}^L$ \cite{Nagaj-JMP-2010}. 
The history state is the ground state of a Hamiltonian that selects configurations for which the work register is updated correctly as the clock proceeds in time. Such a Hamiltonian has the form
\begin{align}
H_{\text{update}} = \sum_{t=1}^L U_t \otimes \ket{t}\bra{t-1}_c + U_t^\dagger \otimes \ket{t-1}\bra{t}_c   .
\end{align}
Whenever the clock undergoes a transition from $\ket{t}$ to $\ket{t+1}$, the
unitary map $U_t$ is applied to the work register. The computation $U_t$ is
undone via the application of $U_t^\dagger$ as the clock transits from
$\ket{t+1}$ to $\ket{t}$. This Hamiltonian is of the form of a quantum walk on
the line of state vectors $\ket{\psi_t} =  (U_t \cdots U_0 \ket{\phi_0}_w )
\otimes \ket{t}_c$. In each state vector $\ket{\psi_t}$ the computation up to step $t$
is stored in the work register as well as a time-marker in the clock register. 
The full Hamiltonian consists of an additional term $H_{\text{input}}$ that select the correct input configuration. Moreover, the highly non-local clock terms can be made strictly local using the construction of a unary clock in which the state vector $\ket{t}$ is represented by an $L$-qubit state vector
\begin{align}
\ket{t'} = \ket{\underbrace{1,\ldots ,1,}_{t} \underbrace{0 ,\ldots ,0}_{L-t} }\, . 
\end{align}
The correct configurations of such a clock are enforced by another additional term $H_{\text{clock}}$. The full Feynman-Kitaev Hamiltonian then reads
\begin{align}
H_{\text{FK}} = H_{\text{input}} + H_{\text{update}} + H_{\text{clock}}\, .\label{fk ham}
\end{align}
The output of the computation can be obtained by measuring the clock register first, and then measuring the work register \footnote{In the spatially local construction of Ref.~\cite{Gosset-PRL-2015} the position measurement of the particle string has the same effect.}. In order to raise the probability of obtaining outcome $L$ upon the clock measurement from $1/(L+1)$ to a constant, one can simply introduce another 
${O}(L)$ identity gates after the last step of the computation. 

The Hamiltonian (\ref{fk ham}) is gapped as $\Delta \propto 1/L^2$, it is
$5$-local (but not spatially local), and frustration-free with vanishing ground state energy, since all terms can be chosen as projections. 
The ground state can be prepared, for example, by adiabatically turning on the
updating term $H_{\text{update}}$ \cite{Bravyi-siam-2010}. If the preparation
procedure starts out in the ground space manifold, and the adiabatic path is
sufficiently smooth the adiabatic theorem tells us that eventually one ends up
in the ground state of the final Hamiltonian $H_{\text{update}}$ in a time that
scales as $1/\Delta^3$ \cite{Jansen-JMP-2007}.  

Finally, note that Proposition~\ref{complexity} implies that arbitrary unitary
and dissipative time-evolution can be simulated efficiently using the ground
states of frustration-free and gapped Hamiltonians on a lattice. This is the
case because arbitrary local Lindbladian dynamics can be approximated by a
poly-length quantum circuit and thus simulated efficiently on a quantum computer \cite{Kliesch-PRL-2011}. 


\section{Towards quantum supremacy}
\label{s:quantum supremacy}

In the previous section we have discussed a general mapping of quantum circuits to unique ground states of frustration-free, poly-gapped local Hamiltonians. In this section we will specialise to the 
family of circuits proposed by Bremner, Montanaro, and Shepherd~\cite{bremner2015average}, for which there is strong complexity-theoretic evidence (in contrast to Shor's algorithm) that their output probability distribution is hard to simulate on any classical computer
within an error of $1/192$ in total variation distance. While this result on its
own is remarkable, one might question whether the ``theoretically hard''
quantum state has actually been prepared in an experiment. However,
using our certification result, one can indeed \emph{certify} the sufficiently precise preparation of a quantum state, such that measurements in the $Z$-basis are intractable for classical computers under a plausible
complexity-theoretic conjecture. This significantly strengthens the perspective
to experimentally demonstrate \emph{quantum supremacy} \cite{preskill2013quantum}.

The result of Ref.~\cite{bremner2015average} is based on the conjectured average-case hardness to compute the \emph{gap} of degree-$3$ polynomials over $\mathbb{F}_2$, $f: \{0,1\}^n \rightarrow 
\{0,1\}$, which can be written (up to an additive constant) as
\begin{equation}
f(x) = \sum_{i,j,k} \alpha_{i,j,k} x_i x_j x_k + \sum_{i,j} \beta_{i,j} x_i x_j + \sum_i \gamma_i x_i \; \text{(mod 2)},
\end{equation}
where $\alpha_{i,j,j},\beta_{i,j},\gamma_i \in \{0,1\}$ and the gap is defined as $gap(f)=|\{x|f(x)=0\}|-|\{x|f(x)=1\}|$. Let the normalised gap be $ngap(f)=gap(f)/2^n$. Then the following is 
assumed to be true.

\begin{conjecture}[\cite{bremner2015average}] \label{conj:hardness}
Let $f: \{0,1\}^n \rightarrow  \{0,1\}$ be a uniformly random degree-3 polynomial over $\mathbb{F}_2$. Then it is $\#\textsf{P}$-hard to approximate $ngap(f)^2$ up to multiplicative error of $1/4+o(1)$ for a $1/24$ fraction of polynomials $f$.
\end{conjecture}

While Conjecture \ref{conj:hardness} is not known to hold for an \emph{uniformly random} choice of polynomials, it is known to be true in the \emph{worst case} \cite{EK90}. All that is left to show is how to lift this result from worst-case to average-case hardness.

The gap as defined above can be expressed conveniently as the acceptance
probability of a family $\mathcal{C}_f$ of so-called IQP circuits on $n$ qubits,
which has the following structure: (i) apply Hadamard gates on all qubits, (ii) apply a sequence of
$Z$, $CZ$, $CCZ$ gates (encoding the terms of $f(x)$), (iii) apply Hadamard gates
on all qubits again. That is, $ngap(f)=\bra{0}^{\otimes n} \mathcal{C}_f \ket{0}^{\otimes n}$. Note for experimental simplicity, that all the diagonal gates commute and the circuit has only constant depth.

We are now ready to state the result of Bremner, Montanaro, and Shepherd precisely:
\begin{theorem}[\cite{bremner2015average}] \label{thm:BMS15}
Assume Conjecture \ref{conj:hardness} is true. If
it is possible to classically sample from the output probability
distribution of any IQP circuit $C$ in polynomial time, up to an
error of $1/192$ in $\ell_1$ norm, then there is a $\textsf{BPP}^{\textsf{NP}}$ algorithm to
solve any problem in $\textsf{P}^{\#\textsf{P}}$. Hence the Polynomial Hierarchy
would collapse to its third level.
\end{theorem}
We would like to note that the Polynomial Hierarchy is of central importance in computational complexity theory
as a generalisation of the better-known $\textsf{P} \neq \textsf{NP}$ conjecture.

Theorem \ref{thm:BMS15} requires to prepare the output distribution of
an IQP circuit with error of $1/192$ in $\ell_1$ norm. 
Using our methods, quantum states closer than $O(1/\mathsf{poly}(n))$ in terms of the system
size $n$ can be certified. Thus sufficiently well-prepared states can be \emph{certified
to be within the provably hard regime} of quantum states that are hard to sample from. 

Without further ado, let us now present the certification procedure proposed.
Let $L$ be the number of gates in an IQP circuit.
\begin{enumerate}
\item Map the IQP circuit to the ground state $\rho_0$ of a frustration-free,
gapped, local Hamiltonian, e.g., according to Proposition \ref{complexity}
and choose $m$ according to Eq.~\eqref{measurements}
\item Flip a fair coin to decide for either (a) or (b):
\begin{enumerate}
	\item Prepare the ground state $\rho_p$ of this local Hamiltonian adiabatically
			  $m$ times, measure local Hamiltonian terms to certify 
				an upper bound on $\norm{\rho_0~-~\rho_p}_1 \leq \varepsilon$, or \label{stepA}
	\item Prepare the ground state $\rho_p$ and measure the position of all spins
			  which projects on the output state of the IQP circuit $\rho_{I}$ with constant
				probability (otherwise, repeat the step), then measure all qubits in the $Z$ basis. \label{stepB}
\end{enumerate}
\end{enumerate}

Step \ref{stepA} certifies, that the quantum simulator indeed prepares a state $\rho_p$ such that $\norm{\rho_0-\rho_p}_1 \leq \varepsilon(m)$.
The ideal ground state $\rho_0$ of the Hamiltonian described in Proposition \ref{complexity} has constant overlap with the output state of the IQP circuit $\rho_I$, i.e., $\norm{\rho_0-\rho_I}_1 \leq c$, which can therefore be projected with constant probability on the subspace where the computation has completed by measuring the respective terms of the Hamiltonian, i.e., 
\begin{equation}
\rho_I = \frac{\Pi \rho_0 \Pi}{\tr(\Pi \rho_0)}. 
\end{equation}
Similarly, we project in step \ref{stepB} the physically prepared
state 
\begin{equation}
	\rho_{I,p} = \frac{\Pi \rho_p \Pi}{\tr(\Pi \rho_p)}.
\end{equation}	
	 It follows by the triangle inequality and properties of the trace distance, that 
\begin{equation}	 
	 \norm{\rho_{I}~-~\rho_{I,p}}_1 \leq \frac{2}{c}\norm{\rho_0-\rho_p}_1 \leq \frac{2}{c}\varepsilon(m), 
\end{equation}		
	which can be driven below the required constant of $1/192$ by an appropriate choice of $m$.
Thus, we can clearly \emph{certify} that a state has been prepared with sufficient precision (i.e., within $1/192$ in total variation distance \footnote{Note that the quantum states being close in trace distance also implies closeness of the measurement distributions.}) such that measuring the first $n$ qubits in the $Z$ basis results in a probability distribution from which it is hard to sample from classically by Theorem~\ref{thm:BMS15}.

Since our ground state certification protocol only applies to states prepared to at least $O(1/\mathsf{poly}(n))$ precision, we note that the protocol proposed above in fact asks for the preparation of states with higher
quality (inverse polynomial accuracy) than what would be strictly necessary to satisfy the hardness result 
(constant accuracy). Nevertheless, we would like to
stress the experimental simplicity to implement reliable local measurements as compared to
more complicated alternatives. 

For example, the experimenter could as well choose to
implement a constant-accuracy fidelity estimation protocol using Hamiltonian simulation and
phase estimation \cite{kitaev1995quantum,NielsenChuang}. Here, the idea is to apply phase
estimation to the exponential $\exp(- \ii H)$ with the prepared state $\rho$
as input.  
Our goal is to estimate 
\begin{equation} 
    \begin{split}
        F(\rho, \rho_0) & = \bra{E_0} \rho \ket{E_0} 
     = \sum_j q_j |\l E_0 | \psi_j
    \r|^2 \\ 
        &=\sum_j q_j F(\proj{\psi_j},\proj{E_0} ) \, , 
    \end{split}
\end{equation} 
where we have written $\rho=\sum_j q_j \ket{\psi_j}\bra{\psi_j}$ in its eigenbasis.
Thus the fidelity of $\rho$, a mixture of pure states, with $\ket{E_0}$ can be equivalently viewed
as a mixture of fidelities of the component pure states $\ket{\psi_j}$ with $\ket{E_0}$.
Therefore it suffices to analyze phase estimation independently for each $\ket{\psi_j}$.
With each input state $\ket{\psi_j} = \sum_{i = 0}^n \psi_{ij} \ket{E_i}$
expressed in the eigenbasis of $H$, the output of the phase estimation
algorithm applied to $\exp(-\ii H)$ is a state of the form $\ket{\phi_j} = \sum_{i =
0}^n p_{ij} \ket{\hat{E}_i} \ket{E_i}$. 
Here, $\ket{\hat{E}_i}$ is a $t$-qubit
state when measured in the computational basis yields the estimate
$\hat{E}_i$ of the eigenvalue $E_i$ of $H$ as its outcome. 
Performing a measurement with POVM elements $\{ \proj{\hat{E}_0}, \mathrm{id} -
\proj{\hat{E}_0} \}$ of the first register then allows one to estimate the desired fidelity 
$F(\proj{\psi_j}, \proj{E_0})$. Recombining the component fidelities according
to their propability weights $q_j$, yields the
overall fidelity $\sum_j q_j F(\proj{\psi_j},\proj{E_0}) = F(\rho,\proj{E_0})$
as the observed frequency of the outcome $E_0$. 

The accuracy of the estimate and the success probability are determined by $t$.
In order to distinguish the different possible measurement outcomes uniquely,
we require that the estimate $\hat{E}_i$ be accurate up to an additive error of
$2^{-n } \ll \Delta$. For this to be the case with probability $1 - \beta$ the
we require the first register to comprise $t = n + \log(2 + 1/2\beta)$ qubits
\cite{NielsenChuang}. Again, using Hoeffding's inequality, we find that in order to
determine $F(\rho, \rho_0) = \sum_j q_j |p_{0j}|^2$
up to an additive error $\epsilon$ with probability $1-
\alpha$, a number $m \geq \ln(2/\alpha)/(2\epsilon^2)$ of i.i.d.\ measurements of
$\proj{\hat{E}_0}$ on $\rho$ are sufficient.

While such an approach would of course remove the precision bottleneck,
requires fewer measurements and
allows the certification of ground states of arbitrary Hamiltonians,
it also puts significantly more requirements on capabilities of a prospective
quantum simulator: it requires an additional $n$-qubit register and the
capability to perform the phase estimation algorithm on the joint system. Essentially, 
a universal quantum computer is required for phase estimation. 
Ultimately, it will be up to the experimentalist to choose the most effective trade-off.


\section{Summary and discussion} \label{discussion}
\label{s:summary}
In this work, we have discussed in what sense weak-membership variants of the certification of
ground states of frustration-free Hamiltonians can be used to certify the correctness of analog quantum simulation
and instances of quantum computation. For this certification to be possible, merely local Hamiltonian terms have to be measured, 
in a way that is perfectly experimentally possible on a number of platforms, in conjunction 
with efficient classical processing. No multi-round interactive proof systems is required for that. The results here challenge the
view that is often expressed specifically in the physics literature, namely that in order to certify a quantum simulator, one needs
to be able to efficiently keep track of the outcome of the simulation. Instead, in important cases, one can certify its
functioning without being able to efficiently predict the actual result of the simulation. The discussion of universal quantum 
computing is in this context is primarily relevant to show that on a classical device computationally hard problems can indeed
be solved by analog quantum simulators.

Having said this, the observations made here further motivate the quest of identifying further intermediate problems. 
Such an intermediate problem would satisfy two criteria: 
\begin{itemize}
\item[(1)] The problem is believed to be intractable classically, but 
it is not necessarily the case that arbitrary quantum computations can be reduced to solving it. 
\item[(2)] A reliable quantum simulator 
is conceivable that can be practically realised with present day experimental capacities and solves that problem. 
\end{itemize}
Identifying 
such problems requires a complexity-theoretic analysis of problems that are related to physical models.
Notions of boson sampling offer interesting instances of intermediate problems \cite{Aaronson,Sampling,AaronsonUniform}.
As an illustration of our approach we have sketched how our method can be combined with the results of Ref.\ \cite{bremner2015average} to certify the preparation of a probability distribution, which is provably 
hard to simulate on any classical computer under a plausible complexity-theoretic assumption.

\section{Acknowledgements}

We would like to thank M.\ Gluza, D.\ Gosset, and B.~M.\  Terhal for discussions 
and the EU (AQuS, SIQS, RAQUEL), the COST network, the DFG (EI 519/7-1), the ERC (TAQ),
the Studienstiftung des Deutschen Volkes, the Templeton Foundation,  
and the Humboldt Foundation for support.
We also thank the anonymous referee to
point out the more complex, but constant-precision alternative for fidelity estimation. 
Finally, we thank Roberto Di Candia for pointing us to a small error in Eqs.\ \eqref{eq:sample complexity} and \eqref{measurements} in an earlier version of this work.



\end{document}